\documentclass[11pt,english,aps,pra,onecolumn,tightenlines,superscriptaddress,notitlepage,floatfix,fleqn,nofootinbib]{revtex4-2}
\renewcommand{\bibliography}[1]{}
\pdfoutput=1
\usepackage[utf8]{inputenc}
\usepackage[T1]{fontenc}
\usepackage{amsmath}
\usepackage{amssymb}
\usepackage{amsfonts}
\usepackage{bm}
\usepackage{bbm}
\usepackage{epsfig}
\usepackage{times}
\usepackage{graphics}

\usepackage[dvipsnames]{xcolor}
\definecolor{dblue}{rgb}{0,0.1,.6}
\definecolor{lightgray}{gray}{0.8}

\newcommand{\todo}[1]  {\textcolor{Purple}{#1}}

\usepackage{setspace}
\usepackage[colorlinks=true,citecolor=dblue,linkcolor=dblue,urlcolor=dblue]{hyperref}
\usepackage[all]{hypcap}
\newcommand{\footnoteref}[1]{\textsuperscript{\ref{#1}}}

\newcommand{\id}{\mathbbm{1}}
\newcommand{\bra}{\langle}
\newcommand{\ket}{\rangle}
\newcommand{\bbra}{\langle\!\langle}
\newcommand{\kket}{\rangle\!\rangle}
\newcommand{\Tr}{\operatorname{Tr}}
\renewcommand{\vec}[1]{{\boldsymbol{#1}}}
\newcommand{\ud}{\mathrm{d}}

\newcommand{\hA}{\hat{A}}
\newcommand{\hB}{\hat{B}}
\newcommand{\hH}{\hat{H}}
\newcommand{\hK}{\hat{K}}
\newcommand{\hL}{\hat{L}}
\newcommand{\hP}{\hat{P}}
\newcommand{\hQ}{\hat{Q}}
\newcommand{\hR}{\hat{R}}
\newcommand{\hS}{\hat{S}}
\newcommand{\hU}{\hat{U}}
\newcommand{\hb}{\hat{b}}
\newcommand{\hc}{\hat{c}}
\newcommand{\hw}{\hat{w}}
\newcommand{\hx}{\hat{x}}
\newcommand{\hz}{\hat{z}}
\newcommand{\hs}{\hat{\sigma}}
\newcommand{\htau}{{\hat{\tau}}}
\newcommand{\dm}{{\hat{\rho}}}
\newcommand{\mri}{\mathrm{i}}

\newcommand{\CC}{\mathbb{C}}
\renewcommand{\Re}{\operatorname{Re}}

\newcommand{\mc}[1]{\mathcal{#1}}
\newcommand{\pdag}{{\vphantom{\dag}}}
\newcommand{\py}{{\vphantom{y}}}

\renewcommand{\L}{\mc{L}}
\newcommand{\K}{\mc{K}}
\newcommand{\J}{\mc{J}}
\newcommand{\D}{\mc{D}}
\newcommand{\E}{\mc{E}}
\renewcommand{\H}{\mc{H}}
\newcommand{\A}{\mc{A}}
\newcommand{\B}{\mc{B}}
\newcommand{\R}{\mc{R}}
\newcommand{\T}{\mc{T}}
\newcommand{\vmu}{\vec{\mu}}
\newcommand{\vnu}{\vec{\nu}}
\renewcommand{\ss}{\text{ss}}

\newcommand  {\Psmatrix}[1]{\left(\begin{smallmatrix}#1\end{smallmatrix}\right)}

\usepackage{amsthm}
\newtheorem{proposition}{Proposition}
\newtheorem{lemma}{Lemma}

\makeatletter
\renewcommand{\p@subsection}{}
\makeatother
\makeatletter
\def\@hangfrom@section#1#2#3{\@hangfrom{#1#2}#3}
\def\@hangfroms@section#1#2{#1#2}
\makeatother

\newcommand{\Emph}[1]{\textbf{\emph{#1}}}

\newcommand{\dart}  {Department of Physics and Astronomy, Dartmouth College, Hanover, New Hampshire 03755, USA}
\newcommand{\duke}  {Department of Physics, Duke University, Durham, North Carolina 27708, USA}
\newcommand{\qlab}  {National Quantum Laboratory, University of Maryland, College Park, MD 20742, USA}
\newcommand{\umd}   {Department of Physics, University of Maryland, College Park, MD 20742, USA}

\begin{document}

\title{\texorpdfstring{A direct algebraic proof for the non-positivity of Liouvillian spectral values\\ and controlling the dissipative gap of unital systems}{A direct algebraic proof for the non-positivity of Liouvillian spectral values and controlling the dissipative gap of unital systems}}
\author{Yikang Zhang}
\affiliation{\dart}
\affiliation{\duke}
\author{Thomas Barthel}
\affiliation{\duke}
\affiliation{\qlab}
\affiliation{\umd}
\date{August 8, 2026}

\begin{abstract}
Markovian open quantum systems are described by the Lindblad master equation $\partial_t\rho =\mathcal{L}(\rho)$, where $\rho$ denotes the system's density operator and $\mathcal{L}$ the Liouville super-operator, which is also known as the Liouvillian. For systems with a finite-dimensional Hilbert space, it is a fundamental property of the Liouvillian that the real parts of all its eigenvalues are non-positive. Analogously, for infinite-dimensional Hilbert spaces, the Liouvillian as a map on trace-class operators only has spectral values with non-positive real parts. The usual arguments for these properties are indirect, using that $\mathcal{L}$ generates a quantum channel and that quantum channels are contractive. We provide a direct algebraic proof based on the Lindblad form of Liouvillians. Subtleties for infinite-dimensional systems are highlighted. For example, unbounded operators can then be eigenvectors of the adjoint Liouvillian with positive eigenvalues, corresponding to a physical instability of the system.
As an illustrative application of the algebraic approach, we prove that a strict minimum can be imposed on the dissipative gap of any unital Markovian spin-$s$ or fermionic many-particle system by adding simple local dissipators. The added dissipators comprise Pauli or generalized Pauli Lindblad operators for spin systems and Majorana Lindblad operators for fermionic systems.
\end{abstract}

\maketitle

\vspace{-3ex}
\begin{spacing}{0.9}
\tableofcontents
\end{spacing}

\section{Introduction}
\subsection{Preliminaries}
In seminal contributions, Lindblad, Gorini, Kossakowski, and Sudarshan showed that the generators $\L_t$ for Markovian quantum dynamics
\begin{equation}\label{eq:LME}
	\partial_t\dm_t=\L_t(\dm_t)
\end{equation}
of open quantum systems, where $\dm_t$ denotes the system's density operator at time $t$, can always be written in the \emph{Lindblad form} \cite{Lindblad1976-48,Gorini1976-17,Breuer2007,Rivas2012,Wolf2008-279}
\begin{equation}\label{eq:Lindblad}
	\partial_t{\dm_t} = \L(\dm_t)=-\mri[\hH,\dm_t]
	 +\sum_a \Big(\hL_a^\pdag \dm_t \hL_a^\dag-\frac{1}{2} \{\hL_a^\dag \hL_a^\pdag,\dm_t\}\Big).
\end{equation}
Here, $\hH$ is the Hamiltonian, comprising the system Hamiltonian and possibly further terms like the Lamb-shift Hamiltonian arising in the weak-coupling regime \cite{Davies1974-39,Davies1976-219,Duemcke1979-34,Breuer2007}, the \emph{Lindblad operators} $\hL_a$ describe the effect of interactions with the environment, $[\hA,\hB]=\hA\hB-\hB\hA$ is the commutator, and $\{\hA,\hB\}=\hA\hB+\hB\hA$ is the anticommutator.
We refer to the generator $\L$ as the \emph{Liouville super-operator} or \emph{Liouvillian}.
To shorten the notation, we have dropped time labels $t$ for the Liouvillian, Hamiltonian, and Lindblad operators in the Lindblad master equation \eqref{eq:LME}. In the following, we will only be concerned with the spectrum of $\L$ at a fixed time.

A first fundamental property of the Liouvillian $\L$ is that its adjoint 
\begin{equation}\label{eq:LindbladAdj}
	\L^\dag(\hA)=\mri[\hH,\hA] + \sum_a\Big(\hL_a^\dag\hA\hL_a^\pdag-\frac{1}{2}\big\{\hL_a^\dag\hL_a^\pdag,\hA\big\}\Big)
\end{equation}
has at least one zero eigenvalue. Here, the adjoint is defined with respect to the trace pairing (Hilbert-Schmidt inner product)
\begin{equation}\label{eq:tracePairing}
	\bbra\hA|\hR\kket:=\Tr(\hA^\dag\hR)\quad \text{such that}\quad
	\bbra\hA|\L(\hR)\kket=:\bbra\L^\dag(\hA)|\hR\kket
\end{equation}
for all operators $\hA,\hR$ on the Hilbert space $\H$ (bounded operators $\A$ and trace-class operators $\hR$ in the case of infinite-dimensional $\H$).
The existence of a zero eigenvalue follows from the trace conservation of the Lindblad master equation \eqref{eq:LME}, which can be shown using the cyclic property $\Tr(\hA\hB)=\Tr(\hB\hA)$ of the trace,
\begin{equation}\label{eq:traceConserv}
	\Tr \L(\hR)=0\ \ \forall\hR
	\quad\stackrel{\eqref{eq:tracePairing}}{\Leftrightarrow}\quad
	\L^\dag(\id)=0.
\end{equation}
It immediately follows that every Markovian open quantum system with a finite-dimensional Hilbert space has at least one steady state $\dm_\ss$ such that $\L(\dm_\ss)=0$, i.e., $\id$ is a left $\L$-eigenvector with eigenvalue zero, and $\dm_\ss$ is a/the corresponding right $\L$-eigenvector with eigenvalue zero.\footnote{\label{Note1}As Liouvillians are super-operators, their eigenvectors are actually operators on the Hilbert space. A system can have multiple steady states such that the zero eigenvalue may be degenerate.}
\begin{figure}[t]
    \centering
    \includegraphics[width=0.85\textwidth]{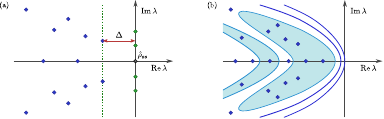}
    \caption{\label{fig:spectrum}\textbf{Structure of Liouvillian spectra.} (a) For finite-dimensional Hilbert spaces, all eigenvalues $\{\lambda\}$ of the Liouvillian have non-positive real parts (Prop.~\ref{prop:nonPositive-finite}), and there is at least one steady state $\dm_\ss$ with $\L$-eigenvalue zero. The spectrum can feature purely imaginary eigenvalues, and, at long times, the system evolves unitarily in the corresponding asymptotic subspace \cite{Baumgartner2008-41,Albert2016-6}. The largest nonzero eigenvalue real part determines the dissipative gap $\Delta$. (b) For infinite-dimensional Hilbert spaces, the spectrum of $\L$ as a map on trace-class operators is still restricted to $\Re\lambda\leq 0$ (Prop.~\ref{prop:nonPositive-inf}) but can have continuous and residual components\todo{\footnoteref{Note5}} as indicated by the lines and shaded regions, and the system need not have a steady state ($\lambda=0$ is always in the spectrum but not necessarily in the point spectrum). Some concrete examples and possible consequences like algebraic (power-law) relaxation are discussed in Refs.~\cite{Prosen2008-10,Cai2013-111,Znidaric2015-92,Song2019-123,Shibata2020-12,LermaHernandez2020-53,Nakagawa2021-126,Haga2021-127,Barthel2020_12,Alaeian2022-5,RubioGarcia2022-106,Ehrhardt2024-6,Mori2024-109}.}
\end{figure}

In this paper, we discuss proofs for a second fundamental property of Liouvillians -- the non-positivity of their spectral values. We point out subtleties for systems with infinite dimensional Hilbert spaces and provide direct algebraic proofs for the non-positivity in Secs.~\ref{sec:viaL-finite} and ~\ref{sec:viaL-inf}. The latter approach can also be useful for bounding dissipative gaps as exemplified for unital systems in Sec.~\ref{sec:normal}.

\subsection{Non-positivity for finite-dimensional Hilbert spaces}
For finite-dimensional Hilbert spaces $\H$, as illustrated in Fig.~\ref{fig:spectrum}, we have:
\begin{proposition}[Non-positivity of Liouvillian eigenvalues for finite $\dim\H$]\label{prop:nonPositive-finite}
	For a Markovian system with a finite-dimensional Hilbert space and Liouvillian $\L$, the real parts of all $\L$-eigenvalues $\lambda_n$ are non-positive,
	\begin{equation}\label{eq:nonPositive-finite}
		\Re\lambda_n\leq 0\quad \forall n.
	\end{equation}
\end{proposition}
In particular, the nonzero real parts correspond to decay rates of excitations towards the steady state (or asymptotic subspace \cite{Baumgartner2008-41,Albert2016-6}) and
\begin{equation}\label{eq:gapFinite}
	\Delta:=-\max_{\Re\lambda_n\neq 0}\Re\lambda_n
\end{equation}
is the so-called \emph{dissipative gap}, which determines the asymptotic decay rate \cite{Buca2012-14,Kessler2012-86,Mori2020-125}. Furthermore, the spectrum is symmetric with respect to the real axis as $\L(\hR)=\lambda\hR$ implies $\L(\hR^\dag)\equiv\big(\L(\hR)\big)^\dag=\lambda^*\hR^\dag$, and the spectra of $\L$ and $\L^\dag$ agree.

\subsection{Non-positivity for infinite-dimensional Hilbert spaces}\label{sec:intro-inf}
For infinite-dimensional Hilbert spaces $\H$, the Liouvillian \eqref{eq:Lindblad} of a Markovian system may be an unbounded super-operator. In such cases, we need to restrict $\L$ to a dense domain $\D(\L)$, i.e., a dense subspace of all trace-class operators $\T(\H)$ such that $\L(\hR)\in\T(\H)$ for all $\hR\in\D(\L)$ and $\overline{\D(\L)}=\T(\H)$.\footnote{\label{Note2}A bounded operator $\hR$ on $\H$ is \emph{trace-class} if and only if $\Tr|\hR|<\infty$ with $|\hR|:=(\hR^\dag\hR)^{1/2}$.} If the Hamiltonian and Lindblad operators are well-behaved physical operators (even unbounded ones like position, momentum, photon number, or bosonic ladder operators), they will generate a closed and densely defined $\L$,\footnote{\label{Note3}$\L:\D(\L)\subseteq \T(\H)\to\T(\H)$ is called \emph{densely defined} if its domain $\D(\L)$ is dense in the space $\T(\H)$ of trace-class operators, i.e., if its closure yields the entire space, $\overline{\D(\L)}=\T(\H)$. It is \emph{closed} if its graph $\{(\hR,\L(\hR))\mid\hR\in\D(\L)\}$ is a closed subset of $\T(\H)\times \T(\H)$, i.e., $\L(\hR)=\hS$ for any sequence $\{\hR_n\in\D(\L)\}$ with $\hR_n\to\hR$ and $\L(\hR_n)\to\hS$.} guaranteeing strong continuity of the associated semigroup $\E_t=e^{\L t}$.\footnote{\label{Note4}A semigroup $\E_t$ is \emph{strongly continuous} if any physical state $\dm$ evolves smoothly at $t=0$ in the sense that $\lim_{t\to 0^+}\|\E_t(\dm)-\dm\|_1=0$.} This is assumed throughout this paper.

Furthermore, for $\dim\H=\infty$, the system need not have a steady state ($\lambda=0$ is not in the point spectrum of $\L$)\footnote{\label{Note5}The spectrum of a Liouvillian generally consists of three parts: The point and continuous spectra as in the hydrogen-atom Hamiltonian, and a residual spectrum. They consist of all $\lambda\in\CC$ for which (a) $(\L-\lambda\id)$ is not injective, (b) $(\L-\lambda\id)$ is injective with a dense range, but the inverse (the resolvent) is unbounded, and (c) $(\L-\lambda\id)$ is injective, but its range is not dense, respectively.} and can be unstable such that $\L^\dag$ has eigenvalues with positive real parts: A simple example is a bosonic mode, where the Lindblad operator $\hL=\sqrt{\gamma}\,\hb^\dag$, given by the creation operator $\hb^\dag$, indefinitely pumps particles into the system at rate $\gamma$, and the adjoint Liouvillian has corresponding positive eigenvalues like $\gamma/2$ and $\gamma$ with
\begin{subequations}\label{eq:example-b}
\begin{alignat}{4}
	\L^\dag(\hb)            &\stackrel{\eqref{eq:LindbladAdj}}{=}\gamma \hb\hb \hb^\dag        &&-\frac{\gamma}{2}\{\hb \hb^\dag,\hb\}&&=\frac{\gamma}{2}\,\hb\ \quad\text{and}\\
	\L^\dag(\hb^\dag\hb+\id)&\stackrel{\eqref{eq:LindbladAdj}}{=}\gamma \hb\hb^\dag\hb \hb^\dag&&-\frac{\gamma}{2}\{\hb \hb^\dag,\hb^\dag\hb\}&&=\gamma\,(\hb^\dag\hb+\id).
\end{alignat}
\end{subequations}
Of course, $\hb$ and $\hb^\dag\hb+\id$ are unbounded operators. As discussed below, the spectrum of $\L^\dag$ is necessarily confined to the left half-plane (non-positive real parts) if we restrict $\L^\dag$ to the domain
\begin{equation}\label{eq:DL+}
	\D(\L^\dag):=\{\hA\in\B(\H)\mid \text{there exists}\ \hB\in\B(\H)\ \ \text{with}\ \ \bbra\hA|\L(\hR)\kket=\bbra\hB|\hR\kket\ \forall\hR\in\D(\L)\},
\end{equation}
which is a weak-$^*$ dense subspace\footnote{\label{Note6}The weak-$^*$ topology is the natural topology for the space of bounded operators $\B(\H)$ when viewed as the dual space of the trace-class operators $\T(\H)$ with the duality induced by the trace pairing \eqref{eq:tracePairing}. A set $\D$ is weak-$^*$ dense in $\B(\H)$ if for every $\hA\in \B(\H)$ there exists a net of operators $\{\hA_\alpha\in \D\}$ which converges to $\hA$ in the sense that $\bbra\hA_\alpha|\hR\kket\to\bbra\hA|\hR\kket$ for all $\hR\in \T(\H)$.} in the space 
of bounded operators $\B(\H)$. In particular, we have:
\begin{proposition}[Non-positivity of Liouvillian spectrum for infinite $\dim\H$]\label{prop:nonPositive-inf}
	For a Markovian system with an infinite-dimensional Hilbert space $\H$, consider the Liouvillian $\L$ as a map on trace-class operators $\T(\H)$ and $\L^\dag$ as a map on bounded operators $\B(\H)$. For unbounded Liouvillians, $\L$ and $\L^\dag$ are restricted to the domains $\D(\L)$ and $\D(\L^\dag)$ as discussed above. Then, all elements of the Liouvillian spectrum $\sigma(\L)$ have non-positive real parts,
	\begin{equation}\label{eq:nonPositive-inf}
		\Re\lambda\leq 0\quad \text{for all}\quad\lambda\in\sigma(\L)=\sigma(\L^\dag)=\{\lambda\in\CC\mid\L-\lambda\id\ \, \text{has no bounded inverse}\}.
	\end{equation}
\end{proposition}
As in the finite-dimensional case, the symmetry of the spectrum $\sigma(\L)$ with respect to the real axis results from the fact that the Liouvillian commutes with the adjoint operation $\A(\hR):=\hR^\dag$. As $\A\L=\L\A$ and $\A^{-1}=\A$, $\L-\lambda^*$ has a bounded inverse if and only if $(\L-\lambda^*)\A=\A(\L-\lambda)$ has an inverse, i.e., if $\lambda\notin \sigma(\L)$. In fact, the symmetry holds individually for the point, continuous, and residual spectra of $\L$.\todo{\footnoteref{Note5}}

As stated in the proposition, the spectra of $\L$ and $\L^\dag$ agree (in general not individually for the different components). This follows from the duality induced by the 
(sesquilinear) trace pairing \eqref{eq:tracePairing}: $\L-\lambda\id$ is invertible if and only if $\L^\dag-\lambda^*\id$ is invertible \cite{Kato1995,Rudin1991}. So 
\begin{equation}\label{eq:spectSymmetry}
	\sigma(\L^\dag)=\sigma(\L)^*=\sigma(\L),
\end{equation}
where the second equality is due to the symmetry of the spectrum.

In light of the example \eqref{eq:example-b}, it is important to note that the non-positivity of the Liouvillian spectrum governs the probabilistic and informational stability of the system, rather than its thermodynamic or energetic stability. The restriction $\Re\lambda\leq 0$ of the rigorous spectrum ensures that the system evolves as a valid quantum channel, where trace distances are contractive and the expectation values of all bounded observables are strictly constrained. However, this does not preclude physical instabilities, such as the continuous particle pumping. Because observables like energy and particle number are unbounded, the system can continuously shift toward infinite energy without violating probability conservation. Consequently, an infinite-dimensional Markovian system can possess a rigorously non-positive Liouvillian spectrum while simultaneously exhibiting diverging expectation values for unbounded observables, cleanly separating the concepts of informational contractivity and energetic stability. See also the discussion in Sec.~\ref{sec:unbounded}.

\subsection{Structure of the paper}
Section~\ref{sec:viaChannel} reviews the conventional indirect argument for the non-positivity (Props.~\ref{prop:nonPositive-finite} and ~\ref{prop:nonPositive-inf}). It employs the fact that $\L$ generates a quantum channel (Prop.~\ref{prop:LtoE}), to then conclude on the non-positivity of the Liouvillian spectral values based on the contractivity of quantum channels (Prop.~\ref{prop:contractivity}).
In Secs.~\ref{sec:viaL-finite} and \ref{sec:viaL-inf}, we provide an alternative entirely algebraic proof for the non-positivity based on the Lindblad form \eqref{eq:Lindblad} of the Liouvillian. Section~\ref{sec:viaL-finite} covers the simpler case of systems with finite-dimensional Hilbert spaces (Prop.~\ref{prop:nonPositive-finite}), and Sec.~\ref{sec:viaL-inf} addresses infinite $\dim\H$ (Prop.~\ref{prop:nonPositive-inf}).
For infinite $\dim\H$, Sec.~\ref{sec:unbounded} enlarges on the example of Eqs.~\eqref{eq:example-b} to discuss the subtle point that, while the Liouvillian spectrum is generally non-positive in the sense of Prop.~\ref{prop:nonPositive-inf}, $\L^\dag$ can have positive eigenvalues on the space of unbounded operators.
As an illustrative application of the algebraic approach, Sec.~\ref{sec:normal} proves that the dissipative gap of any unital Markovian spin-$s$ or fermionic many-particle system can be opened and rigorously bounded from below by adding certain single-site dissipators -- using Pauli or generalized Pauli Lindblad operators for spin systems and Majorana Lindblad operators for fermionic systems. We also provide interesting counterexamples for non-unital systems, where the additional single-site dissipators reduce or even close the gap. One such counterexample is the infinite-range transverse-field Ising model with local spontaneous emission, where the additional dissipators can induce criticality.
The paper concludes with the discussion section~\ref{sec:discuss}.

\section{Non-positivity through properties of quantum channels (the conventional argument)}\label{sec:viaChannel}
\begin{proposition}[Liouvillians generate quantum channels \cite{Lindblad1976-48,Gorini1976-17}]\label{prop:LtoE}
	Every Liouvillian \eqref{eq:Lindblad} generates a quantum channel, i.e., a completely positive trace-preserving (CPTP) map\footnote{\label{Note7}In the quantum channel $\E_t=e^{\L t}$, $\L$ is a time-instantaneous Liouvillian $\L\equiv\L_{t_0}$ for any time $t_0$ of interest. So, for systems with time-dependent Liouvillians $\L_t$, $\E_t(\dm)$ does \emph{not} solve the Lindblad master equation \eqref{eq:LME}. $\E_t$ is only considered to assess the spectrum of $\L_{t_0}$.}
	\begin{equation}\label{eq:Et}
		\E_t=e^{\L t}\quad\forall t\geq 0.
	\end{equation}
\end{proposition}
\begin{proof}
	To see this, note that $\E_t$ is trace-preserving because
	\begin{equation}
		\partial_t\Tr\E_t(\hR)=\Tr \L(\E_t(\hR))\stackrel{\eqref{eq:traceConserv}}{=}0.
	\end{equation}
	To see that $\E_t$ is also completely positive (CP), we can split the Liouvillian \eqref{eq:Lindblad} into a sum of two terms
	\begin{gather}\label{eq:Lsplit}
		\L=\K+\J\quad\text{with}\quad
		\K(\hR):=-\hK\hR - \hR\hK^\dag,\quad
		\J(\hR):=\sum_a\hL_a^\pdag\hR\hL_a^\dag,
	\end{gather}
	and $\hK=\mri\hH+\frac{1}{2}\sum_a\hL_a^\dag\hL_a^\pdag$. Evidently, $e^{\K t}(\hR)=e^{-\hK t}\hR e^{-\hK^\dag t}$ and $e^{\J t}=\sum_k \frac{t^k}{k!}\J^k(\hR)$ are CP. The Lie-Trotter product formula \cite{Trotter1959} $\E_t=e^{\K t +\J t}=\lim_{N\to\infty}\left(e^{\K t/N }e^{\J t/N }\right)^N$ then implies that $\E_t$ is also CP. The CP maps form a closed convex cone in the space of superoperators such that the CP property is retained under the limit $N\to\infty$.
\end{proof}

\begin{proposition}[Quantum channels are contractive \cite{Nielsen2000}]\label{prop:contractivity}
	For a system with a unique Hilbert space $\H$,\footnote{\label{Note8}For simplicity, we assume that the system has a unique Hilbert space such that the decomposition of $\dm-\dm'$ into its positive and negative spectral components $\hQ_\pm$ is possible. Otherwise, we would need to work in the $C^*$-algebra framework, where physical quantities are represented as elements of a unital $C^*$-algebra $\mc{A}$, states are linear positive functionals on $\mc{A}$, and the dynamics is described by unital CP maps like $e^{\L^\dag t}$ \cite{Lindblad1976-48,Davies1976,Bratteli1987-1}. Examples for which this can be necessary are infinite-size systems with spontaneous symmetry breaking and, according to Haag's theorem, relativistic quantum field theory \cite{Haag1955-29,Hall1957-31}.} the trace distance \cite{Bhatia1997} of density operators $\dm$ and $\dm'$ is non-increasing under the action of any quantum channel $\E$, i.e.,
	\begin{equation}\label{eq:contractivity}
		\|\E(\dm)-\E(\dm')\|_1\leq \|\dm-\dm'\|_1\quad\forall \dm,\dm'\in\T(\H).
	\end{equation}
\end{proposition}
\begin{proof}
	Let us decompose $\dm-\dm'$ into its positive and negative spectral components such that $\dm-\dm'=\hQ_+-\hQ_-$ with positive semidefinite operators $\hQ_\pm\succeq 0$. Note that $\Tr\hQ_+=\Tr\hQ_-$ because $\Tr(\dm-\dm')=0$. Hence,
	\begin{equation}\label{eq:traceDist}
		\|\dm-\dm'\|_1\equiv\Tr|\dm-\dm'|
		=\Tr(\hQ_++\hQ_-)=2\Tr\hQ_+.
	\end{equation}
	Further, let $\hP_+$ denote the projection onto the positive spectral subspace of $\E(\hQ_+)-\E(\hQ_-)=\E(\dm-\dm')$. Then,
	\begin{align*}
		\|\dm-\dm'\|_1
		&\stackrel{\eqref{eq:traceDist}}{=} 2\Tr\hQ_+=2\Tr\E(\hQ_+)\\
		&\geq 2\Tr\big[\hP_+\E(\hQ_+)\big]
		 \geq 2\Tr\big[\hP_+\big(\E(\hQ_+)-\E(\hQ_-)\big)\big]\\
		&=2\Tr\big[\hP_+\E(\dm-\dm')\big]=\|\E(\dm)-\E(\dm')\|_1,
	\end{align*}
	where we have also used that $\E$ is trace preserving and $\E(\hQ_\pm)\succeq 0$.
\end{proof}

\Emph{Finite-dimensional Hilbert spaces.} ---
For finite $\dim\H$, the contractivity \eqref{eq:contractivity} of the quantum channel \eqref{eq:Et} implies that all eigenvalues $\epsilon_n$ of $\E_t=e^{\L t}$ lie in the unit disk of the complex plane,
\begin{equation}\label{eq:EtEigen}
	\E_t(\hR_n)=\epsilon_n\hR_n\quad\Rightarrow\quad
	|\epsilon_n|=\frac{\|\E_t(\hR_n)\|_1}{\|\hR_n\|_1} \stackrel{\eqref{eq:contractivity}}{\leq} 1.
\end{equation}
The eigenvalues of the Liouvillian $\L$ are $\lambda_n=\frac{1}{t}\ln \epsilon_n$, and Eq.~\eqref{eq:EtEigen} implies that their real parts are all non-positive as stated in Prop.~\ref{prop:nonPositive-finite}.

\Emph{Infinite-dimensional Hilbert spaces.} ---
Similarly, for systems with a unique infinite-dimensional Hilbert space,\todo{\footnoteref{Note8}}
the contractivity (Prop.~\ref{prop:contractivity}) of the strongly continuous semigroup\todo{\footnoteref{Note4}} $e^{\L t}$ (Prop.~\ref{prop:LtoE}) implies that the Liouvillian spectrum $\sigma(\L)$ [Eq.~\eqref{eq:nonPositive-inf}] is confined to the left half of the complex plane (non-positive real parts) \cite{Engel2000}: For any complex number $\lambda$ with $\Re\lambda>0$, we can explicitly construct the resolvent $(\L-\lambda\id)^{-1}$ as the Laplace transform $\R_\lambda:=-\int_0^\infty \ud t\, e^{-\lambda t} \E_t$. Due to the contractivity of the quantum channel,
\begin{equation}
	\|\E_t\|_{1\to 1}:=\sup_{\hR\neq 0\in \D(\L)} \frac{\|\E_t(\hR)\|_1}{\|\hR\|_1} \stackrel{\eqref{eq:contractivity}}{\leq} 1
	\quad\Rightarrow\quad
	\|\R_\lambda\|_{1\to 1}\leq \int_0^\infty \ud t\, |e^{-\lambda t}|=\frac{1}{\Re\lambda},
\end{equation}
i.e., $\R_\lambda$ converges to a well-defined, bounded super-operator, and it is indeed the inverse of $\L-\lambda\id$ as
\begin{equation}
	(\L-\lambda\id)\R_\lambda=-\int_0^\infty \ud t\, (\L-\lambda\id)\,e^{-\lambda t}\E_t
	=-\int_0^\infty \ud t\, \frac{\ud}{\ud t}\left(e^{-\lambda t} \E_t\right)
	=\id.
\end{equation}
In the first equality, we are allowed to pull $\L$ into the integral because it is closed and densely defined.\todo{\footnoteref{Note3}} The last equality follows from the fundamental theorem of calculus. So, for any $\Re\lambda>0$, the resolvent $(\L-\lambda\id)^{-1}$ is a bounded super-operator such that $\lambda$ is \emph{not} in the spectrum of $\L$ as stated in Prop.~\ref{prop:nonPositive-inf}.\todo{\footnoteref{Note5}}

\section{Non-positivity via algebraic properties of the Liouvillian for finite \texorpdfstring{$\dim\H$}{dim(H)}}\label{sec:viaL-finite}
While the usual argument from Sec.~\ref{sec:viaChannel} is clear and uses important properties of quantum channels and Liouvillians, one would hope to also be able to conclude on the non-positivity (Prop.~\ref{prop:nonPositive-finite}) directly from the Lindblad form \eqref{eq:Lindblad} of the Liouvillian, only using algebraic properties. We will see in the following that this is indeed possible.

We start by proving two useful properties of Liouvillians, which also apply for infinite $\dim\H$ and will be reused in Sec.~\ref{sec:viaL-inf}. The first is one of the Kossakowski conditions \cite{Kossakowski1972-3,Rivas2012}.
\begin{lemma}\label{prop:lemma1-finite}
	Let $\{|i\ket\}$ be an orthonormal basis for a system with a separable Hilbert space. Transition rates between distinct orthogonal states are non-negative, i.e.,
	matrix elements of the Liouvillian \eqref{eq:Lindblad} obey,
	\begin{equation}\label{eq:np-1}
		\bra j|\L\big(|i\ket\bra i|\big)|j\ket\geq 0\quad \forall i\neq j.
	\end{equation}
\end{lemma}
\begin{proof}
	The Lindblad form \eqref{eq:Lindblad} of the Liouvillian implies
	\begin{align}\nonumber 
		\bra j|\L\big(|i\ket\bra i|\big)|j\ket
		=&-\mri\bra j|\hH|i\ket\bra i|j\ket+\mri\bra j|i\ket\bra i|\hH|j\ket \\\nonumber
		&+\sum_a\Big(\bra j|\hL_a^\pdag|i\ket\bra i|\hL_a^\dag|j\ket
		 -\frac{1}{2}\bra j|\hL_a^\dag\hL_a^\pdag|i\ket\bra i|j\ket
		 -\frac{1}{2}\bra j|i\ket\bra i|\hL_a^\dag\hL_a^\pdag|j\ket\Big) \\
		=&\sum_a\Big(\big|\bra j|\hL_a|i\ket\big|^2
		            -\bra i|\hL_a^\dag\hL_a^\pdag|i\ket\,\delta_{i,j}\Big).
		\label{eq:matrixEl1}
	\end{align}
 The matrix element \eqref{eq:matrixEl1} is real and, for all $i\neq j$,
	$\bra j|\L\big(|i\ket\bra i|\big)|j\ket=\sum_a\big|\bra j|\hL_a|i\ket\big|^2\geq 0$, proving
	the non-negativity \eqref{eq:np-1}.
\end{proof}

\begin{lemma}\label{prop:lemma2}
	For any Liouvillian $\L$ and operator $\hA$, we have the partial ordering
	\begin{equation}\label{eq:LAordering}
		\L^\dag(\hA^\dag\hA)\succeq \L^\dag(\hA^\dag)\hA+\hA^\dag\L^\dag(\hA).
	\end{equation}
\end{lemma}
\begin{proof}
	The positive semidefiniteness of the operator
	\begin{align}
		\L^\dag(&\hA^\dag\hA)- \L^\dag(\hA^\dag)\hA-\hA^\dag\L^\dag(\hA) \nonumber \\
		\stackrel{\eqref{eq:LindbladAdj}}{=}&
		\mri [\hH,\hA^\dag\hA]-\mri[\hH,\hA^\dag]\hA-\mri \hA^\dag[\hH,\hA]
		+\sum_a\Big( \hL_a^\dag\hA^\dag\hA\hL_a^\pdag-\frac{1}{2}\hL_a^\dag\hL_a\hA^\dag\hA-\frac{1}{2}\hA^\dag\hA\hL_a^\dag\hL_a \nonumber\\
		&-\hL_a^\dag\hA^\dag\hL_a^\pdag\hA+\frac{1}{2}\hL_a^\dag\hL_a^\pdag\hA^\dag\hA+\frac{1}{2}\hA^\dag\hL_a^\dag\hL_a^\pdag\hA-\hA^\dag\hL_a^\dag\hA\hL_a^\pdag+\frac{1}{2}\hA^\dag\hL_a^\dag\hL_a^\pdag\hA+\frac{1}{2}\hA^\dag\hA\hL_a^\dag\hL_a^\pdag \Big) \nonumber \\
		=&\sum_a\Big( \hL_a^\dag\hA^\dag\hA\hL_a^\pdag-\hL_a^\dag\hA^\dag\hL_a^\pdag\hA-\hA^\dag\hL_a^\dag\hA\hL_a^\pdag+ \hA^\dag\hL_a^\dag\hL_a^\pdag\hA\Big) \nonumber\\
		=&\sum_a \Big(\hL_a\hA-\hA\hL_a \Big)^\dag\Big(\hL_a\hA-\hA\hL_a \Big)\succeq 0
		\label{eq:inequ_AA}
	\end{align}
	is equivalent to the partial ordering \eqref{eq:LAordering}.
\end{proof}

Based on these two lemmas, the non-positivity of the $\L$-eigenvalue real parts for systems with a finite-dimensional Hilbert space (Prop.~\ref{prop:nonPositive-finite}) can be established as follows: Suppose that $\L$ has an eigenvalue $\lambda$. Then $\lambda$ is also an eigenvalue of $\L^\dag$.\footnote{\label{Note9}Note that this is in general only true for finite-dimensional Hilbert spaces. For example, the Liouvillian for the model in Eq.~\eqref{eq:example-b} has a left eigenvector $\id$ with eigenvalue $0$ (trace conservation), but no steady state, i.e., no corresponding right eigenvector.} Let $\hA$ denote the corresponding eigenvector of $\L^\dag$ such that
\begin{equation}\label{eq:Leig}
	\L^\dag(\hA)=\lambda\hA \quad\text{and}\quad
	\L^\dag(\hA^\dag)=(\L^\dag(\hA))^\dag=\lambda^*\hA^\dag.
\end{equation}
Employing Lemma~\ref{prop:lemma2}, we have
\begin{equation}\label{eq:ineq-re}
	\L^\dag(\hA^\dag\hA)\stackrel{\eqref{eq:LAordering}}{\succeq} \L^\dag(\hA^\dag)\hA+\hA^\dag\L^\dag(\hA)
	\stackrel{\eqref{eq:Leig}}{=}(\lambda^*+\lambda)\hA^\dag\hA=2\Re(\lambda)\hA^\dag\hA.
\end{equation}
The operator $\hA^\dag\hA$ is positive semidefinite and can be diagonalized such that
\begin{equation}\label{eq:AAdiag}
	\hA^\dag\hA=\sum_i a_i | i\ket\bra i|,\quad \text{where}\quad
	\bra i|j\ket=\delta_{i,j}.
\end{equation}
As $\hA\neq 0$ and Prop.~\ref{prop:nonPositive-finite} concerns finite-dimensional Hilbert spaces, we can renormalize $\hA$ by a constant and order the eigenvalues of $\hA^\dag\hA$ such that
\begin{equation}\label{eq:a_i}
	1=a_1\geq a_2\geq\cdots\geq0.
\end{equation}
Then, using Eq.~\eqref{eq:ineq-re} we have
\begin{equation}\label{eq:ReBound}
	2\Re(\lambda)=2\Re(\lambda)\bra 1|\hA^\dag\hA|1\ket
	\stackrel{\eqref{eq:ineq-re}}{\leq} \bra 1|\L^\dag\big(\hA^\dag\hA\big)|1\ket.
\end{equation}
Finally, with Lemma~\ref{prop:lemma1-finite} we find
\begin{align}\nonumber
	\bra 1|\L^\dag\big(\hA^\dag\hA\big)|1\ket
	&\stackrel{\eqref{eq:AAdiag}}{=}\sum_i a_i\bra 1|\L^\dag\big(|i\ket\bra i|\big)|1\ket
	= \bra 1|\L^\dag\big(|1\ket\bra 1|\big)|1\ket+\sum_{i>1} a_i \bra 1|\L^\dag\big(|i\ket\bra i|\big)|1\ket\\
    &\hspace{-1.5ex}\stackrel{\eqref{eq:np-1},\eqref{eq:a_i}}{\leq}\! \bra 1|\L^\dag\big(|1\ket\bra 1|\big)|1\ket+\sum_{i>1} \bra 1|\L^\dag\big(|i\ket\bra i|\big)|1\ket
    =\bra 1|\L^\dag\big(\id\big)|1\ket
	\stackrel{\eqref{eq:traceConserv}}{=}0.
    \label{eq:final}
\end{align}
The inequalities~\eqref{eq:ReBound} and \eqref{eq:final} show that $\Re\lambda\leq 0$, i.e., we have found an entirely algebraic proof for the fundamental Prop.~\ref{prop:nonPositive-finite}, only using the Lindblad form \eqref{eq:Lindblad} of the Liouvillian.

\section{Non-positivity via algebraic properties of the Liouvillian for infinite \texorpdfstring{$\dim\H$}{dim(H)}}\label{sec:viaL-inf}
For the algebraic proof of non-positivity for infinite-dimensional Hilbert spaces (Prop.~\ref{prop:nonPositive-inf}), we can reformulate Lemma~\ref{prop:lemma1-finite} as:
\begin{lemma}\label{prop:lemma1-inf}
	Let $\dm\in\D(\L)$ and $\htau\in\D(\L^\dag)$ be self-adjoint positive semidefinite operators with disjoint support, i.e., $\bbra\htau|\dm\kket=0$. The corresponding transition rates according to the dynamics generated by the Liouvillian \eqref{eq:Lindblad} are non-negative in the sense that
	\begin{equation}\label{eq:np-2}
		\Tr\big(\htau\L(\dm)\big)\stackrel{\eqref{eq:tracePairing}}{\equiv} \bbra \L^\dag(\htau)|\dm\kket \geq 0\quad \text{if}\quad \bbra\htau|\dm\kket=0.
	\end{equation}
\end{lemma}

Using the equality \eqref{eq:spectSymmetry} of the $\L$ and $\L^\dag$ spectra, we need to show that all $\lambda\in\CC$ with $\Re\lambda>0$ are in the resolvent set $r(\L^\dag)$ of $\L^\dag$ as an operator on the weak-$^*$ dense subspace $\D(\L^\dag)$ of bounded operators,\todo{\footnoteref{Note6}} defined in Eq.~\eqref{eq:DL+}.

For any bounded operator $\htau\in\D(\L^\dag)$ with $\htau\succeq 0$, there exists a state $\dm\in\T(\H)$ or at least a sequence $\{\dm_n\}$ of trace-class operators with the limit $\dm$ as a state on the $C^*$-algebra of bounded operators \cite{Bratteli1987-1} such that
\begin{equation}\label{eq:normAchievingRho}
	\bbra\htau|\dm\kket = \|\htau\|\quad\text{with}\quad \dm\succeq 0\quad\text{and}\quad \Tr(\dm)=1.
\end{equation}
The norm-achieving state $\dm$ is the equivalent of $|1\ket\bra 1|$ in Sec.~\ref{sec:viaL-finite}. Now, $\htau':=\|\htau\|\,\id-\htau$ is positive semidefinite and obeys $\bbra\htau'|\dm\kket=0$ such that Lemma~\ref{prop:lemma1-inf} gives
\begin{equation}\label{eq:np-3}
	0\stackrel{\eqref{eq:np-2}}{\geq} - \bbra \L^\dag(\htau')|\dm\kket \stackrel{\eqref{eq:traceConserv}}{=}  \bbra \L^\dag(\htau)|\dm\kket,
\end{equation}
where we have used the trace conservation \eqref{eq:traceConserv}. The inequality \eqref{eq:np-3} shows that any Markovian time evolution tends to lead out of the norm-achieving subspace of $\htau$.
Choosing $\htau=\hA^\dag\hA$ for any bounded operator $\hA\in\D(\L^\dag)$, Lemma~\ref{prop:lemma2} yields the bound
\begin{equation}\label{eq:np-4}
	0\stackrel{\eqref{eq:np-3}}{\geq} \bbra \L^\dag(\hA^\dag\hA)|\dm\kket
	\stackrel{\eqref{eq:LAordering}}{\geq} \bbra\L^\dag(\hA^\dag)\hA|\dm\kket+\bbra\hA^\dag\L^\dag(\hA)|\dm\kket
	=2\Re\bbra\hA^\dag\L^\dag(\hA)|\dm\kket.
\end{equation}

To assess the resolvent set of $\L^\dag$, consider $\hB:=(\lambda \id-\L^\dag)\hA=\lambda\hA-\L^\dag(\hA)$. With
\begin{equation}\label{eq:np-5}
	\Re\bbra\hA^\dag\hB|\dm\kket=\Re\bbra\hA^\dag(\lambda \id-\L^\dag)\hA|\dm\kket
	\stackrel{\eqref{eq:normAchievingRho}}{=}\Re\lambda\|\hA\|^2 - \Re\bbra\hA^\dag\L^\dag(\hA)|\dm\kket
	\stackrel{\eqref{eq:np-4}}{\geq} \Re\lambda\,\|\hA\|^2,
\end{equation}
the upper bound $\Re\bbra\hA^\dag\hB|\dm\kket\leq |\bbra\hA^\dag\hB|\dm\kket|=|\Tr(\hB\dm\hA^\dag)|$,
and the Cauchy-Schwarz inequality for pre-inner products \cite{Rudin1991,Reed1980}
\begin{equation}
	|\Tr(\hB\dm\hA^\dag)|^2\leq |\Tr(\hA\dm\hA^\dag)|\,|\Tr(\hB\dm\hB^\dag)| \leq \|\hA\|^2\|\hB\|^2,
\end{equation}
we arrive at the bound
\begin{equation}\label{eq:np-6}
	\|\hA\|\|\hB\|\stackrel{\eqref{eq:np-5}}{\geq}\Re\lambda\,\|\hA\|^2
	\quad\Leftrightarrow\quad
	\|(\L^\dag-\lambda \id)\hA\|\geq \Re\lambda\,\|\hA\|,
\end{equation}
which ensures that $\L^\dag-\lambda \id$ is injective with a closed range if $\Re\lambda>0$. Because $\L^\dag$ generates a physical (Markovian) semigroup, it follows from the Lumer-Phillips theorem that the range is dense \cite{Engel2000}. Thus, $\L^\dag-\lambda \id$ has a bounded everywhere-defined inverse if $\Re\lambda>0$, i.e., $\{\lambda\in\CC\mid \Re \lambda> 0\} \subseteq r(\L^\dag)$.
As an alternative to using the Lumer-Phillips, we can employ a resolvent-continuation argument: 
By assumption, $\L$ generates a strongly continuous semigroup\todo{\footnoteref{Note4}} $\E_t=e^{\L t}$, which necessarily obeys $\|\E_t\|_{1\to 1}\leq M e^{\omega t}$ for some $M\geq 1$ and finite $\omega$. Hence, every $\lambda$ with sufficiently large real part belongs to the resolvent set $r(\L)$ of the Liouvillian, and hence to $r(\L^\dag)$ such that the set $\Omega:=r(\L^\dag)\cap\{\lambda\mid\Re\lambda>0\}$ is nonempty and open. Due to Eq.~\eqref{eq:np-6}, $\Omega$ is also relatively closed in the right half-plane which, due to connectedness, implies $\Omega=\{\lambda\mid\Re\lambda>0\}$.

This completes the proof of Prop.~\ref{prop:nonPositive-inf}.

\section{Instability and positive eigenvalues for unbounded operators}\label{sec:unbounded}
As noted in the introduction (Sec.~\ref{sec:intro-inf}), infinite-dimensional systems can exhibit physical instabilities, such as continuous particle pumping. At first glance, this may appear to contradict the non-positivity of the Liouvillian spectrum (Prop.~\ref{prop:nonPositive-inf}). To clarify this subtlety, we must carefully distinguish between formal algebraic eigenoperators and the true spectrum of the adjoint Liouvillian acting on the Banach space of bounded operators.

Consider the unstable bosonic mode introduced above Eqs.~\eqref{eq:example-b}, where the system is driven by a single pumping Lindblad operator $\hL=\sqrt{\gamma}\,\hb^\dag$. The corresponding Liouvillian and its formal adjoint are
\begin{equation}
    \L(\dm) = \gamma \Big( \hb^\dag \dm \hb - \frac{1}{2} \{\hb\hb^\dag, \dm\} \Big) \quad \text{and} \quad
    \L^\dag(\hA) = \gamma \Big( \hb \hA \hb^\dag - \frac{1}{2} \{\hb\hb^\dag, \hA\} \Big).
\end{equation}
As pointed out in Eqs.~\eqref{eq:example-b}, the annihilation operator $\hb$ and the shifted number operator $\hb^\dag\hb + \id$ are examples of formal $\L^\dag$ eigenoperators with strictly positive eigenvalues $\gamma/2$ and $\gamma$.
Physically, these positive algebraic eigenvalues dictate the unstable dynamics of the system's observables. The equation of motion $\partial_t\bra\hA \ket_t = \bra\L^\dag(\hA)\ket_t$ leads to exponentially growing expectation values
\begin{equation}
	\bra \hb \ket_t         = e^{\gamma t / 2} \bra\hb\ket_0\quad\text{and}\quad
	\bra \hb^\dag\hb \ket_t = e^{\gamma t}     \big(\bra \hb^\dag\hb \ket_0+1\big) - 1.
\end{equation}
The system experiences a divergence in both its field amplitude and its particle number (energy).

Despite these positive formal eigenvalues, this does not violate Prop.~\ref{prop:nonPositive-inf}. The resolution to this apparent paradox lies in the strict domain definitions required for infinite-dimensional spaces. The statement $\sigma(\L) = \sigma(\L^\dag) \subset \{z \in \CC \mid \Re z \le 0\}$ rigorously holds when $\L$ acts on the Banach space of trace-class operators $\T(\H)$ and $\L^\dag$ acts on its dual space of bounded operators $\B(\H)$. The domain $\D(\L^\dag)$ defined in Eq.~\eqref{eq:DL+} is a subspace of $\B(\H)$, and the unbounded operators $\hb$ and $\hb^\dag\hb+1$ reside outside this domain.
Consequently, the values $\gamma/2$ and $\gamma$ are merely roots of the formal algebraic expression for the adjoint Liouvillian evaluated on an extended algebra of unbounded operators, but they are \emph{not} part of the rigorous spectrum $\sigma(\L^\dag)$. The latter is defined via the boundedness of the resolvent $(\L^\dag - \lambda \id)^{-1}$ acting on the specific space $\D(\L^\dag)\subset\B(\H)$ rather than just finding algebraic roots as in Eqs.~\eqref{eq:example-b}.

The necessity of this careful distinction becomes especially clear if one attempts to regularize the instability via a finite truncation of the Hilbert space. Suppose we truncate the bosonic mode to a maximum of $N$ particles, $\H_N =\operatorname{span}\{|0\ket,|1\ket,\dotsc,|N\ket\}$. In this finite-dimensional space, the truncated creation operator $\hb^\dag_N$ is bounded, and the system is covered by Prop.~\ref{prop:nonPositive-finite}. Due to the truncation, the system can no longer pump particles indefinitely, but rather the population piles up in the state $|N\ket\bra N|$, which becomes a valid steady state ($\lambda = 0$). In this truncated space, the exact algebraic relations of Eqs.~\eqref{eq:example-b} are broken due to boundary effects at $|N\ket$, the positive formal eigenvalues instantly vanish, and the entire spectrum collapses back into the left half-plane $\Re\lambda\leq 0$. Unboundedness of the operators is the essential mathematical feature required to support diverging expectation values.

As pointed out at the end of Sec.~\ref{sec:intro-inf}, the non-positivity of the rigorous Liouvillian spectrum (Prop.~\ref{prop:nonPositive-inf}) guarantees that the system is informationally and probabilistically stable; $\dm_t$ remains a valid quantum state. It does not guarantee that the system is energetically stable; the state is free to travel to infinity.

Finally, the functional-analytic perspective clarifies why an infinite-dimensional Markovian system may fail to possess a normalizable steady state despite being strictly trace-preserving. Algebraically, trace preservation is expressed as $\L^\dag(\id)=0$, meaning the identity is a left eigenoperator of the Liouvillian with eigenvalue zero. In finite-dimensional systems, this automatically guarantees the existence of a corresponding right eigenoperator $\dm_\ss$ in the point spectrum of $\L$, such that $\L(\dm_\ss)=0$. However, in infinite dimensions, the identity resides in the dual space of bounded operators $\B(\H)$. Its existence as a left eigenoperator does not necessitate a corresponding right eigenoperator in the predual space of trace-class operators $\T(\H)$; instead, $\lambda=0$ may merely belong to the residual spectrum of $\L$.\todo{\footnoteref{Note5}} Physically, this occurs when the probability distribution ``escapes to infinity.'' In the pumped bosonic mode, for instance, the total probability $\Tr\dm_t$ is strictly conserved at all times as $\L^\dag(\id)=0$. Yet, because the Lindblad operator $\hL=\sqrt{\gamma}\,\hb^\dag$ continuously injects particles, the probability mass shifts indefinitely toward higher Fock states. For example,
\begin{equation}
	\dm_0=|0\ket\bra 0|\quad\Rightarrow\quad
	\dm_t=\sum_{n=0}^\infty e^{-\gamma t} \big(1-e^{-\gamma t}\big)^n |n\ket\bra n|.
\end{equation}
The population of any finite-energy subspace asymptotically decays to zero, the system never equilibrates into a stationary probability distribution, and no valid trace-class density operator $\dm_\ss$ with $\L(\dm_\ss)=0$ exists.

\section{Opening the dissipative gap of unital systems}\label{sec:normal}
The algebraic framework used to prove the non-positivity of Liouvillian spectra in Secs.~\ref{sec:viaL-finite} and \ref{sec:viaL-inf} can also be used to bound the dissipative gap \eqref{eq:gapFinite}. Note that for infinite $\dim\H$, the definition of the dissipative gap $\Delta$ needs to be adapted such that
\begin{equation}\label{eq:gapInf}
	\Delta := -\sup \left\{ \Re\lambda \mid \lambda\in\sigma(\L)\setminus\sigma_0(\L) \right\},
\end{equation}
where $\sigma_0(\L)$ is the asymptotic spectrum of the Liouvillian, defined as the set of all spectral values that lie exactly on the imaginary axis ($\Re\lambda=0$). This includes steady state(s) at $\lambda= 0$, as well as any purely imaginary eigenvalues corresponding to decoherence-free subspaces or rotating asymptotic phases.

As specific examples, we will show in the following how to impose a strict minimum on the dissipative gap of any \emph{unital} Markovian spin-$s$ and fermionic many-particle system by adding certain simple single-site dissipators. A system is unital if the maximally mixed state $\id$ is a steady state. This is for example the case if all Lindblad operators are normal ($\hL_a^\pdag\hL_a^\dag=\hL_a^\dag\hL_a^\pdag$),
\begin{equation}\label{eq:normalLind}
	\L(\id)=0\quad\Leftrightarrow\quad \sum_a\hL_a^\pdag\hL_a^\dag=\sum_a\hL_a^\dag\hL_a^\pdag.
\end{equation}

We will first prove this for finite-size systems in Secs.~\ref{sec:normal-qubit}-\ref{sec:normal-s}, then comment on the generalization to the thermodynamic limit in Sec.~\ref{sec:normal-tdLimit}, and finally give counterexamples for nonunital systems to substantiate the importance of the unitality constraint. The first counterexample (Sec.~\ref{sec:normalCounter-finite}) is a system of three spins-$1/2$, where addition of the gap-opening dissipators actually reduces the gap. The second counterexample (Sec.~\ref{sec:normalCounter-Ising}) is the infinite-range transverse Ising model with spontaneous emission, where the additional dissipators can push the model to the critical point and, hence, close the gap.

\subsection{Spin-\texorpdfstring{$1/2$}{1/2} systems}\label{sec:normal-qubit}
\begin{proposition}[Controlling the gap of unital spin-$1/2$ systems]\label{prop:normal-qubit}
	Consider a Markovian system \eqref{eq:Lindblad} with unital dynamics \eqref{eq:normalLind}, consisting of $N$ spins-$1/2$. If we add the dissipator
	\begin{equation}\label{eq:normal-DkP}
		\D_\kappa(\dm):=\kappa\sum_{j=1}^N\sum_{\mu=x,y}\left(\hs_j^\mu\dm\hs_j^\mu-\dm\right)
	\end{equation}
	where $\{\hs_j^{x\py},\hs_j^y,\hs_j^{z\py}\}$ are the Pauli operators for spin $j$, the system has a unique steady state with a dissipative gap $\Delta\geq 2\kappa\geq 0$.
\end{proposition}
\begin{proof}
Let $\L_0$ denote the Liouvillian of the original system and $\L:=\L_0+\D_\kappa$ the Liouvillian of the modified system.
As $\L_0$ and $\D_\kappa$ are both unital, the modified system has the identity $\id$ (the infinite-temperature state) as a steady state. It is the \emph{unique} steady state, because the added Lindblad operators $\{\hs_j^{x\py},\hs_j^y\mid j=1,\dotsc,N\}$ generate the full associative operator $\CC$-algebra such that the system is \emph{Davies irreducible} according to Theorem~5.2 in Ref.~\cite{Zhang2024-57} -- hence, only having the trivial invariant subspaces \cite{Fagnola2002-43,Zhang2024-57} and a unique full-rank steady state \cite{Evans1978-2-17,Umanita2005-134,Zhang2024-57}. See also Ref.~\cite{Yoshida2024-109}.

Now, as in the non-positivity proof of Sec.~\ref{sec:viaL-finite}, let $\hA$ be an excited eigenvector of the total adjoint Liouvillian such that
\begin{equation}\label{eq:normalRe}
	\L^\dag(\hA) = \lambda \hA\quad\Rightarrow\quad
	\Re \Tr\big(\hA^\dag\L^\dag(\hA)\big) = \Re(\lambda) \Tr(\hA^\dag\hA).
\end{equation}
For the $\L_0$ term on the left side of the latter equation, Lemma~\ref{prop:lemma2} implies
\begin{align}\nonumber
	2\Re \Tr\big(\hA^\dag\L_0^\dag(\hA)\big)
	&=\Tr\big(\L_0^\dag(\hA^\dag)\hA + \hA^\dag\L_0^\dag(\hA)\big)
	\stackrel{\eqref{eq:LAordering}}{\leq}\Tr\big(\L_0^\dag(\hA^\dag\hA)\big)\\
	&\stackrel{\eqref{eq:tracePairing}}=\bbra \L_0^\dag(\hA^\dag\hA) | \id\kket
	\equiv \bbra \hA^\dag\hA | \L_0(\id)\kket=0,
\end{align}
where we have used that the infinite temperature state $\id$ is a steady state of $\L_0$.
With this, Eq.~\eqref{eq:normalRe} gives
\begin{equation}\label{eq:normalRe2}
	\Re(\lambda) \Tr(\hA^\dag\hA)\leq \Re \Tr\big(\hA^\dag\D_\kappa^\dag(\hA)\big).
\end{equation}

To assess $\Tr\big(\hA^\dag\D_\kappa^\dag(\hA)\big)$, we can expand $\hA$ in terms of Pauli strings
\begin{equation}\label{eq:PauliString}
	\hP_\vmu:=\bigotimes_{j=1}^N \hs^{\mu_j}\quad\text{with}\quad
	\mu_j\in\{0,x,y,z\}\quad\text{and}\quad \hs^0=\id.
\end{equation}
Importantly, the excited $\L^\dag$ eigenvector $\hA$ is traceless (has no identity component with $\mu_j=0$ $\forall j$), because it is necessarily orthogonal to the steady state, $\bbra\hA|\id\kket=\Tr(\hA^\dag)=0$.
The Pauli strings are mutually orthogonal, $\bbra\hP_\vmu|\hP_{\vmu'}\kket=2^N\delta_{\vmu,\vmu'}$, and are all eigenvectors of $\D_\kappa^\dag$. Specifically, if $N_\vmu^{x\py}$, $N_\vmu^y$, and $N_\vmu^{z\py}$ are the numbers of $x$, $y$, and $z$ operators in the Pauli string $\hP_\vmu$,
\begin{equation}\label{eq:normal-DP}
	\D_\kappa^\dag\big(\hP_\vmu\big)=-2\kappa\left(N_\vmu^{x\py}+N_\vmu^y+2N_\vmu^{z\py}\right)\hP_\vmu
\end{equation}
such that, with $N_\vmu^{x\py}+N_\vmu^y+N_\vmu^{z\py}\geq 1$ in the expansion of $\hA$,
\begin{equation}
	\Re(\lambda) \Tr(\hA^\dag\hA)
	\stackrel{\eqref{eq:normalRe2}}{\leq} \Re \Tr\big(\hA^\dag\D_\kappa^\dag(\hA)\big)
	\stackrel{\eqref{eq:normal-DP}}{\leq} -2\kappa \Tr(\hA^\dag\hA).
\end{equation}
This implies the bound $\Delta\geq 2\kappa$ for the dissipative gap.
\end{proof}

\subsection{Fermionic systems}\label{sec:normal-fermion}
\begin{proposition}[Controlling the gap of unital fermionic systems]\label{prop:normal-fermion}
	Consider a Markovian system \eqref{eq:Lindblad} with unital dynamics \eqref{eq:normalLind}, consisting of $N$ fermionic modes with annihilation operators $\hc_j$ and the associated $2N$ self-adjoint Majorana operators $\hw_{2j-1}:=\hc_j^\pdag+\hc_j^\dag$ and $\hw_{2j}:=-\mri(\hc_j^\pdag-\hc_j^\dag)$. If we add the dissipator
	\begin{equation}
		\D_\kappa(\dm):=\kappa\sum_{k=1}^{2N}\left(\hw_k\dm\hw_k-\dm\right),
	\end{equation}
	the system has a unique steady state with a dissipative gap $\Delta\geq 2\kappa\geq 0$.
\end{proposition}
\begin{proof}
Closely following the proof for the spin-$1/2$ case in Sec.~\ref{sec:normal-qubit}, $\L_0$ shall denote the Liouvillian of the original system, $\L:=\L_0+\D_\kappa$ the Liouvillian of the modified system, and $\hA$ an excited $\L^\dag$-eigenvector with $\L^\dag(\hA) = \lambda \hA$. Again, the identity $\id$ is the unique steady state as $\L_0$ and $\D_\kappa$ are both unital and the added Majorana Lindblad operators generate the full associative operator $\CC$-algebra \cite{Zhang2024-57}.

To assess $\Tr\big(\hA^\dag\D_\kappa^\dag(\hA)\big)$ in Eq.~\eqref{eq:normalRe2}, we can expand $\hA$ in terms of Majorana strings, which parallel the Pauli strings \eqref{eq:PauliString},
\begin{equation}\label{eq:MajoranaString}
	\hP_\vmu:=\prod_{k=1}^{2N} \hw_k^{\mu_k}\quad\text{with}\quad
	\mu_k\in\{0,1\},
\end{equation}
where the identity $\id$ corresponds to $\mu_k=0$ for all $k$. Because $\hA$ is necessarily orthogonal to the steady state, it is traceless and therefore contains no identity component, i.e., $|\vmu| := \sum_{k=1}^{2N} \mu_k \geq 1$ for all terms in the expansion.

The Majorana strings are mutually orthogonal and act as exact eigenvectors of $\D_\kappa^\dag$. By conjugating a string $\hP_\vmu$ with a single Majorana operator $\hw_j$, the string acquires a sign depending on whether it commutes or anticommutes. Since $\hw_j$ anticommutes with all other Majorana operators ($\{\hw_j,\hw_k\}=2\delta_{j,k}$), it anticommutes with $\hP_\vmu$ if and only if the number of \emph{other} operators in the string, given by $|\vmu|-\mu_j$, is odd. Thus, $\hw_j\hP_\vmu\hw_j = (-1)^{|\vmu|-\mu_j}\hP_\vmu$, and
\begin{equation}\label{eq:normal-DP-fermion}
	\D_\kappa^\dag\big(\hP_\vmu\big) = \kappa \sum_{k=1}^{2N} \left( (-1)^{|\vmu|-\mu_k} - 1 \right) \hP_\vmu.
\end{equation}
We can evaluate this for any non-identity string ($|\vmu| \geq 1$):
\begin{itemize}
  \item If $|\vmu|$ is odd, the $2N-|\vmu|$ terms with $\mu_k=0$ yield $(-1)^{|\vmu|-\mu_k}-1=-2$, and the remaining terms give $0$. So, the eigenvalue is $-2\kappa(2N-|\vmu|)\leq -2\kappa$.
  \item If $|\vmu|$ is even, the $|\vmu|$ terms with $\mu_k=1$ yield $(-1)^{|\vmu|-\mu_k}-1=-2$, and the remaining terms give $0$. So, the eigenvalue is $-2\kappa|\vmu|\leq-4\kappa$.
\end{itemize}
Hence, the maximum eigenvalue for a non-identity string is $-2\kappa$, and we have again $\Re \Tr\big(\hA^\dag\D_\kappa^\dag(\hA)\big) \leq -2\kappa \Tr(\hA^\dag\hA)$. With this, Eq.~\eqref{eq:normalRe2} implies the bound $\Delta\geq 2\kappa$ for the dissipative gap.
\end{proof}

\subsection{Spin-\texorpdfstring{$s$}{s} systems}\label{sec:normal-s}
To generalize from the spin-$1/2$ case in Sec.~\ref{sec:normal-qubit} to general spin-$s$, we can employ Sylvester's generalized Pauli matrices (a.k.a.\ Weyl-Heisenberg matrices) \cite{Sylvester1904,Weyl1931,Schwinger1960-46}. They are defined as
\begin{subequations}\label{eq:Sylvester}
\begin{gather}
	\hw_{\mu,\nu} := \hx^\mu \hz^\nu
	= \sum_{k=0}^{d-1} \omega^{\nu k} |k+\mu\bmod d\ket\bra k|.
	\quad\text{for}\quad
	\mu,\nu \in \{0,\dotsc,d-1\}\\
	\text{with}\quad d:=2s+1 \quad\text{and}\quad \omega := e^{2\pi \mri/d},\quad\text{where}\\\label{eq:shiftClock}
	\hx|k\ket = |k+1\bmod d\ket \quad\text{and}\quad
	\hz|k\ket = \omega^k|k\ket
\end{gather}
define the shift operator $\hx$ and clock operator $\hz$. The operators $\{\hw_{\mu,\nu}\}$ are unitary (and hence normal), and mutually orthogonal with
\begin{equation}
	\hz\hx = \omega \,\hx\hz\quad\text{and}\quad
	\bbra\hw_{\mu,\nu}|\hw_{\mu',\nu'}\kket=d\,\delta_{\mu,\mu'}\delta_{\nu,\nu'}.
\end{equation}
\end{subequations}

\begin{proposition}[Controlling the gap of unital spin-$s$ systems]\label{prop:normal-s}
	Consider a Markovian system \eqref{eq:Lindblad} with unital dynamics \eqref{eq:normalLind}, consisting of $N$ spins-$s$ with local Hilbert space dimension $d=2s+1$. If we add the dissipator
	\begin{equation}
		\D_\kappa(\dm):=\kappa\sum_{j=1}^N \left( \hx_j^\pdag\dm\hx_j^\dag + \hx_j^\dag\dm\hx_j^\pdag + \hz_j^\pdag\dm\hz_j^\dag + \hz_j^\dag\dm\hz_j^\pdag - 4\dm \right)
	\end{equation}
	where $\hx_j$ and $\hz_j$ are Sylvester's shift and clock operators \eqref{eq:shiftClock} for site $j$, the system has a unique steady state with a dissipative gap $\Delta\geq 2\kappa\big(1-\cos\frac{2\pi}{d}\big)\geq 0$.
\end{proposition}
\begin{proof}
Closely following the proofs for the spin-$1/2$ and fermionic cases in Secs.~\ref{sec:normal-qubit} and \ref{sec:normal-fermion}, $\L_0$ shall denote the Liouvillian of the original system, $\L:=\L_0+\D_\kappa$ the Liouvillian of the modified system, and $\hA$ an excited $\L^\dag$-eigenvector with $\L^\dag(\hA) = \lambda \hA$. The identity $\id$ is the unique steady state because $\L_0$ and $\D_\kappa$ are both unital, and the added generalized-Pauli Lindblad operators generate the full associative operator $\CC$-algebra, making the system \emph{Davies irreducible} \cite{Zhang2024-57}.

To assess $\Tr\big(\hA^\dag\D_\kappa^\dag(\hA)\big)$ in Eq.~\eqref{eq:normalRe2}, we can expand $\hA$ in terms of generalized Pauli strings
\begin{equation}\label{eq:GenPauliString}
	\hP_{\vmu,\vnu} := \bigotimes_{j=1}^N \hw_{\mu_j,\nu_j} \quad\text{with}\quad \mu_j,\nu_j \in \{0, \dotsc, d-1\},
\end{equation}
where the identity $\id$ corresponds to $\mu_j=\nu_j=0$ for all $j$. Because $\hA$ is necessarily orthogonal to the steady state, it is traceless and therefore contains no identity component in its expansion.

The generalized Pauli strings are mutually orthogonal and act as exact eigenvectors of $\D_\kappa^\dag$. Conjugating a local operator $\hw_{\mu,\nu}$ by the Lindblad operators $\hx$ and $\hz$ simply yields a phase:
\begin{equation}
	\hx \hw_{\mu,\nu} \hx^\dag \stackrel{\eqref{eq:Sylvester}}{=} \omega^{-\nu} \hw_{\mu,\nu}\quad\text{and}\quad
	\hz \hw_{\mu,\nu} \hz^\dag \stackrel{\eqref{eq:Sylvester}}{=} \omega^{\mu} \hw_{\mu,\nu}.
\end{equation}
Hence, application of the single-site dissipator to $\hw_{\mu,\nu}$ gives
\begin{equation}
	\kappa\,(\omega^\nu + \omega^{-\nu} + \omega^\mu + \omega^{-\mu} - 4)\hw_{\mu,\nu}
	= 2\kappa\left( \cos\frac{2\pi\nu}{d} + \cos\frac{2\pi\mu}{d} - 2 \right)\hw_{\mu,\nu}.
\end{equation}
Thus,
\begin{equation}\label{eq:normal-DGenPauli}
	\D_\kappa^\dag\big(\hP_{\vmu,\vnu}\big) = 2\kappa \sum_{j=1}^N \left( \cos\frac{2\pi \nu_j}{d} + \cos\frac{2\pi \mu_j}{d} - 2 \right) \hP_{\vmu,\vnu}.
\end{equation}
We can evaluate this for any non-identity string, where at least one site $j$ has $(\mu_j,\nu_j) \neq (0,0)$. The maximum possible eigenvalue for such a non-identity string occurs when exactly one index on one site is $1$ or $d-1$ and all other indices are $0$, giving an eigenvalue of $2\kappa\big(\cos\frac{2\pi}{d} - 1\big)$.

As this maximum eigenvalue for a non-identity string is strictly negative, we have $\Re \Tr\big(\hA^\dag\D_\kappa^\dag(\hA)\big) \leq -2\kappa\big(1 - \cos\frac{2\pi}{d}\big) \Tr(\hA^\dag\hA)$. With this, Eq.~\eqref{eq:normalRe2} implies the bound $\Delta \geq 2\kappa\big(1 - \cos\frac{2\pi}{d}\big)$ for the gap.
\end{proof}

\subsection{Generalization to the thermodynamic limit}\label{sec:normal-tdLimit}
The obtained lower bounds on the dissipative gap in Props.~\ref{prop:normal-qubit}-\ref{prop:normal-s} are independent of the system size $N$. In order to rigorously adapt the proofs for the thermodynamic limit $N \to \infty$, we only need to address one important mathematical subtlety -- the divergence of the standard trace. The divergence of $\Tr(\id) = d^N$ means that the Hilbert-Schmidt inner product \eqref{eq:tracePairing} is no longer applicable: The eigenvectors $\hA$ of $\L^\dag$ as well as the Pauli and Majorana strings, which we use to expand $\hA$, are bounded operators, but they have infinite trace norms in the thermodynamic limit. For example, every Pauli string $\hP_{\vmu}$ squares to the identity ($\hP_{\vmu}^\dag \hP_{\vmu} = \id$).
Furthermore, as noted in Prop.~\ref{prop:nonPositive-inf}, the spectrum of $\L^\dag$ may develop continuous or residual components, meaning an exact, bounded eigenvector $\hA$ for a given spectral value $\lambda$ may not exist.

To rigorously extend the gap bounds to infinite systems, we study the dynamics on the quasi-local $C^*$-algebra of the infinite lattice -- the \emph{uniformly hyperfinite} (UHF) algebra for spins or the \emph{canonical anticommutation relation} (CAR) algebra for fermions \cite{Bratteli1987-1,Bratteli1997-2}. We replace the unnormalized trace with the unique, \emph{normalized tracial state} $\tau$, which represents the infinite-temperature state. It is defined as
\begin{equation}
	\tau(\hA) = \lim_{N\to\infty}\, \frac{1}{d^N} \Tr_N(\hA)
\end{equation}
with $d=2$ for spins-$1/2$ and fermions, and $d=2s+1$ for spins-$s$. This functional satisfies $\tau(\id) = 1$ and $\tau(\hA\hB) = \tau(\hB\hA)$, and it induces a rigorous inner product on the Gelfand-Naimark-Segal (GNS) Hilbert space of observables \cite{Gelfand1943-12,Segal1947-53,Bratteli1987-1},
\begin{equation}
    \bbra \hA | \hB \kket_\tau := \tau(\hA^\dag \hB).
\end{equation}

Because $\L=\L_0+\D_\kappa$ is unital, the identity remains the steady state (the tracial state is invariant), implying $\tau(\L^\dag(\hA)) = 0$ for all observables $\hA$ and all $\kappa$.

To bound the dissipative gap $\Delta$, we must assess the spectral values $\lambda$ that govern the slowest relaxation rates. A fundamental result in functional analysis guarantees that the topological boundary of the spectrum is always contained in the approximate point spectrum \cite{Conway1990,Kato1995}, $\partial\sigma(\L^\dag) \subset \sigma_{\text{ap}}(\L^\dag)$. Because the dissipative gap \eqref{eq:gapInf} is defined via the supremum of the real parts of the spectrum, the gap is strictly determined by spectral values on this boundary. Therefore, for every spectral value $\lambda \in \partial\sigma(\L^\dag)$ on the boundary of the spectrum, Weyl's criterion applies \cite{Kato1995}: There exists a Weyl sequence of approximate $\L^\dag$ eigenvectors $\{\hA_n\}$ with unit norm 
\begin{equation}\label{eq:WeylSeq}
    \tau(\hA_n^\dag \hA_n^\pdag) = 1\ \forall n\quad\text{and}\quad
    \lim_{n\to\infty} \|(\L^\dag - \lambda\id)\hA_n\|_\tau = 0.
\end{equation}
For any non-zero spectral value ($\lambda \neq 0$), we can project this sequence such that it is orthogonal to the steady state, ensuring $\tau(\hA_n) = 0$ for all $n$.

Applying the tracial state $\tau$ to the ordering relation \eqref{eq:LAordering} in Lemma~\ref{prop:lemma2}, in analogy to the finite-size case, we find
\begin{equation}
    0=\tau\left( \L_0^\dag(\hA_n^\dag\hA_n^\pdag)\right) \stackrel{\eqref{eq:LAordering}}{\geq} 2 \Re \tau\big(\hA_n^\dag \L_0^\dag(\hA_n^\pdag)\big),
\end{equation}
where the left equality follows from $\tau\circ\L_0^\dag=0$.

In accordance with Eq.~\eqref{eq:normalRe2}, we are left to assess $\tau\big(\hA_n^\dag\D_\kappa^\dag(\hA_n)\big)$ for the gap-opening dissipator $\D_\kappa^\dag$. The local Pauli, Majorana, or generalized Pauli strings form a strictly orthonormal basis with respect to the $\tau$ inner product, e.g., $\tau(\hP_{\vmu}^\dag \hP_{\vec{\nu}}) = \delta_{\vmu, \vec{\nu}}$. Because $\tau(\hA_n) = 0$, the operators $\hA_n$ consist entirely of non-identity strings. As derived in the finite-size proofs, $\D_\kappa^\dag$ acts diagonally on these bases with a strictly negative maximum eigenvalue $-c\kappa$ for all non-identity strings. Thus, we strictly have
\begin{equation}
    \Re \tau\big(\hA_n^\dag \D_\kappa^\dag(\hA_n^\pdag)\big)
    \leq -c\kappa \tau(\hA_n^\dag \hA_n^\pdag)
    \stackrel{\eqref{eq:WeylSeq}}{=} -c\kappa.
\end{equation}

Finally, evaluating the inner product of the Weyl sequence with the full adjoint Liouvillian $\L^\dag = \L_0^\dag + \D_\kappa^\dag$ yields the gap bound
\begin{equation}
    \Re(\lambda) \stackrel{\eqref{eq:WeylSeq}}{=} \lim_{n\to\infty} \Re \tau\big(\hA_n^\dag \L^\dag(\hA_n^\pdag)\big)
    = \lim_{n\to\infty} \left[ \Re \tau\big(\hA_n^\dag \L_0^\dag(\hA_n^\pdag)\big) + \Re \tau\big(\hA_n^\dag \D_\kappa^\dag(\hA_n^\pdag)\big) \right] 
    \leq 0 - c\kappa.
\end{equation}
This rigorously demonstrates that the algebraic bounds derived from Lemma~\ref{prop:lemma2} universally constrain the decay rates associated with the boundary of the spectrum, thereby establishing the lower bound $\Delta\geq c\kappa$ for the dissipative gap \eqref{eq:gapInf} in the thermodynamic limit.

\subsection{Counterexample -- Incomplete gap opening in a finite non-unital system}\label{sec:normalCounter-finite}
To demonstrate the necessity of the unitality constraint \eqref{eq:normalLind} in Prop.~\ref{prop:normal-qubit}, we construct a finite-dimensional counterexample using a system of $N=3$ spins-$1/2$. The computational basis states are denoted by $\{|i\ket\mid i=0,\dotsc,7\}$, where the binary representation of $i$ signifies the states of the three qubits.

Consider the Liouvillian $\L_0$ featuring highly asymmetric, non-normal Lindblad operators $|i\ket\bra j\neq i|$,
\begin{equation}
	\L_0(\dm)=\sum_{i,j=0,\, i\neq j}^7\gamma_{i,j}\D[|i\ket\bra j|](\dm)\quad\text{with}\quad
	\D[\hL](\dm):=\hL \dm \hL^\dag-\frac{1}{2} \{\hL^\dag \hL,\dm\}.
\end{equation}
To engineer a system with a degenerate zero eigenvalue, we choose the transition rate matrix $\gamma$ to possess two disconnected invariant subspaces:
\begin{equation}
	\gamma = \Psmatrix{
		0 & 1 & 0 & 0 & 0 & 0 & 0 & 0 \\
		100 & 0 & 0 & 100 & 0 & 100 & 0 & 0 \\
		0 & 0 & 0 & 0 & 0 & 0 & 0 & 0 \\
		0 & 0 & 0 & 0 & 0 & 0 & 0 & 0 \\
		0 & 0 & 0 & 0 & 0 & 0 & 0 & 0 \\
		0 & 0 & 0 & 0 & 0 & 0 & 0 & 0 \\
		0 & 0 & 100 & 0 & 100 & 0 & 0 & 100 \\
		0 & 0 & 0 & 0 & 0 & 0 & 0 & 0}.
\end{equation}
The jump operators $|i\ket\bra j|$ are not normal ($\hL_a^\pdag\hL_a^\dag \neq \hL_a^\dag\hL_a^\pdag$), and the system violates the precondition \eqref{eq:normalLind} of Prop.~\ref{prop:normal-qubit}. 

We can analytically determine the dissipative gap $\Delta_0$ of this unperturbed Liouvillian. Because all Lindblad operators are of the form $|i\ket\bra j|$, the dynamics of the populations $\rho_i:=\bra i|\dm|i\ket$ completely decouple from the coherences $\rho_{i,j\neq i}:=\bra i|\dm|j\ket$,
\begin{equation}
	\partial_t\rho_i \stackrel{\eqref{eq:Lindblad}}{=} \sum_{j,\, j \neq i} \gamma_{ij} \rho_j - \Gamma_i \rho_i\quad\text{and}\quad
	\partial_t\rho_{i,j} \stackrel{\eqref{eq:Lindblad}}{=} - \frac{1}{2} (\Gamma_i + \Gamma_j) \rho_{i,j},\quad\text{where}\quad
	\Gamma_j := \sum_i \gamma_{i,j}
\end{equation}
is the total decay rate out of state $|j\ket$.
As $\Gamma_6=0$, state $|6\ket$ is a dark (steady) state. Concurrently, the states $\{|0\ket, |1\ket\}$ form a closed population loop: State $|0\ket$ decays exclusively to $|1\ket$ with $\Gamma_0 = 100$, and $|1\ket$ decays back to $|0\ket$ with $\Gamma_1 = 1$. The eigenvalues of this isolated $2 \times 2$ population block are exactly $0$ and $-101$. The remaining states $|2\ket$, $|3\ket$, $|4\ket$, $|5\ket$, and $|7\ket$ are transient and isolated from one another, each decaying into one of the two invariant subspaces at a rate of $\Gamma_j = 100$. Thus, we have exactly two steady states ($|6\ket$ and the steady state of the invariant subspace $\operatorname{span}\{|0\ket, |1\ket\}$), and the slowest decaying population modes have the rate $100$.
However, the coherences $\rho_{i,j\neq i}$ decay at the rate $\frac{1}{2}(\Gamma_i + \Gamma_j)$. The absolute slowest decaying mode of the entire system is the coherence between the two most stable states, $|1\ket$ and $|6\ket$, which yields a dissipative gap \eqref{eq:gapFinite} of $\Delta_0 = \frac{1}{2}(1 + 0) = 1/2$.

We now attempt to increase the gap of the system by adding the auxiliary dissipator $\D_1$ from Eq.~\eqref{eq:normal-DkP} with strength $\kappa=1$. For a unital system, Prop.~\ref{prop:normal-qubit} would algebraically guarantee an increased dissipative gap $\Delta\geq 2\kappa = 2$. However, numerical diagonalization of the combined Liouvillian $\L = \L_0 + \D_1$ yields $\Delta \approx 0.448$.

The fact that $\Delta$ is strictly smaller than $2\kappa=2$ proves that the macroscopic gap bound can be violated if the original system contains non-unital dissipation. The massive off-diagonal mixing generated by the highly asymmetric, non-normal Lindblad operators successfully repels the eigenvalues toward the imaginary axis, dominating the homogenizing effect of $\D_1$ and underscoring the necessity of the unitality constraint.

In this example, adding $\D_\kappa$ did not only fail to achieve $\Delta\geq 2\kappa$, it actually lowered the gap from $0.5$ to $0.448$, while also breaking the steady-state degeneracy of the original system.

\subsection{Counterexample -- Gap closing in the dissipative infinite-range transverse Ising model}\label{sec:normalCounter-Ising}
To further illustrate the failure of the algebraic gap bounds in the presence of non-unital dissipation, we consider a macroscopic many-body system in the thermodynamic limit ($N\to\infty$) -- the infinite-range transverse-field Ising model.

We construct the unperturbed Liouvillian $\L_0$ using an $XX$ interaction, a transverse field $g$ along the $z$ axis, and local spontaneous emission dissipators
\begin{equation}
	\L_0(\dm) = -\mri[\hH, \dm] + \gamma \sum_{i=1}^N \D[\hs_i^-](\dm) \quad\text{with}\quad
	\hH = -\frac{J}{N} \sum_{i<j} \hs_i^x \hs_j^x + g \sum_{i=1}^N \hs_i^z.
\end{equation}
The Lindblad operators $\hs_i^- = \frac{1}{2}(\hs_i^{x\py} - \mri\hs_i^y)$ cause non-unital dissipation.
$\L_0$ possesses a global $\mathbb{Z}_2$ parity symmetry generated by $\hU_\text{parity} = \prod_i \hs_i^z$. As $\hU_\text{parity}$ commutes with the Hamiltonian and all Lindblad operators change sign under parity conjugation‌, this is a weak symmetry.
The system also features a weak symmetry under permutations $\hU_\pi$ with $\hU_\pi\hs_i^\mu\hU_\pi^\dag=\hs_{\pi(i)}^\mu$ and $\pi\in S_N$. Recall that a \emph{weak symmetry} \cite{Baumgartner2008-41,Buca2012-14,Albert2014-89} with respect to a unitary transformation $\hU$ is defined by the property $\L(\hU\dm\hU^\dag) =\hU\L(\dm)\hU^\dag$.

We attempt to impose a minimum dissipative gap on the system by adding the auxiliary dissipator $\D_\kappa$ from Eq.~\eqref{eq:normal-DkP}, which is also invariant under the $\mathbb{Z}_2$ and permutation symmetries. If $\L_0$ was unital, Prop.~\ref{prop:normal-qubit} would guarantee a gap $\Delta \ge 2\kappa$.

To assess the true gap of the combined Liouvillian $\L = \L_0 + \D_\kappa$, we evaluate the mean-field dynamics. In the thermodynamic limit, the mean-field decoupling $\bra \hs_i^x \hs_j^x \ket \approx \bra \hs_i^x \ket \bra \hs_j^x \ket$ generally becomes exact due to the all-to-all coupling \cite{Spohn1980-52,Mori2013-06}. Defining the single-spin magnetizations $m^\mu := \bra \hs_i^\mu \ket$ and exploiting the permutation symmetry, the equations of motion $\partial_t m^\mu = \bra\L^\dag(\hs_i^\mu)\ket_t$ evaluate to
\begin{subequations}
\begin{align}
    \partial_t m^x &= -2g m^y - \Gamma_{xy} m^x, \\
    \partial_t m^y &= 2(J m^z + g) m^x - \Gamma_{xy} m^y, \\
    \partial_t m^z &= -2J m^x m^y - \Gamma_z m^z - \gamma,
\end{align}
\end{subequations}
where the effective decay rates are augmented by the auxiliary dissipator such that $\Gamma_{xy} = \frac{\gamma}{2} + 2\kappa$ and $\Gamma_z = \gamma + 4\kappa$.

Because the dynamics preserve the $\mathbb{Z}_2$ symmetry, the system admits a trivial paramagnetic steady state where the $x$ and $y$ magnetizations vanish ($m^x = m^y = 0$). Setting $\partial_t m^z = 0$ yields the steady-state $z$ magnetization
\begin{equation}
    m^z_\ss = -\frac{\gamma}{\Gamma_z} = -\frac{\gamma}{\gamma + 4\kappa}.
\end{equation}

The dissipative gap \eqref{eq:gapInf} of the system in the paramagnetic phase is determined by the stability of this steady state against small fluctuations $\delta m^\mu$. Linearizing the equations of motion around the paramagnetic solution $\vec{m}=(0,0,m^z_\ss)$, the transverse sector decouples with the Jacobian matrix $M$:
\begin{equation}
    M = \begin{pmatrix} -\Gamma_{xy} & -2g \\ 2g + 2J m^z_\ss & -\Gamma_{xy} \end{pmatrix}.
\end{equation}
The eigenvalues of $M$ govern the relaxation rates of the $x$-$y$ excitations. If the interaction $J$ is sufficiently strong, the symmetric steady state becomes unstable, triggering spontaneous symmetry breaking into the ferromagnetic phase. The gap of the combined Liouvillian $\L_0 + \D_\kappa$ closes at the non-equilibrium critical point $J=J_c(\kappa)$, which is determined by
\begin{equation}
    \det(M) = 0\quad\Leftrightarrow\quad
    \Gamma_{xy}^2 + 2g(2g + 2J_c m^z_\ss) = 0,
\end{equation}
giving
\begin{equation}
    J_c(\kappa) 
    = -\frac{1}{4g m^z_\ss}\left(\Gamma_{xy}^2 + 4g^2 \right)
    = \frac{\gamma + 4\kappa}{4g \gamma} \left[ \left(\frac{\gamma}{2} + 2\kappa\right)^2 + 4g^2 \right].
\end{equation}
At this critical point, the system does not merely exhibit a single vanishing eigenvalue but develops a true continuous spectrum of gapless excitations, similar to the illustration in Fig.~\ref{fig:spectrum}b:
In the permutation-symmetric subspace $\mathbb{L}_\text{sym}$ of the full $4^N$ dimensional Liouville space $\mathbb{L}$, the total spin operator $\sum_i\hat{\vec{\sigma}}_i$ effectively behaves as a single macroscopic degree of freedom evolving in an effective potential well. At criticality, the local curvature of this effective potential vanishes, eliminating the restoring force for macroscopic amplitude fluctuations. For $N \to \infty$, the entire discrete tower of collective excited states compresses. The energy spacings between these modes scale algebraically to zero, forming a continuous gapless spectrum at $\lambda = 0$.
\begin{figure}[t]
    \centering
    \includegraphics[width=0.48\textwidth]{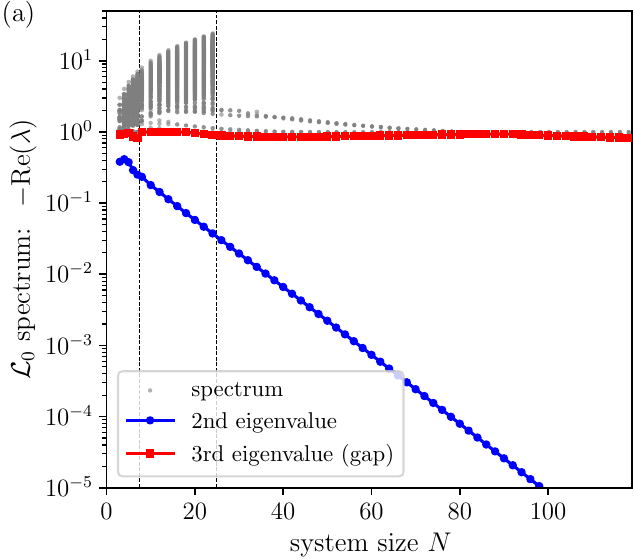}\quad\ 
    \includegraphics[width=0.48\textwidth]{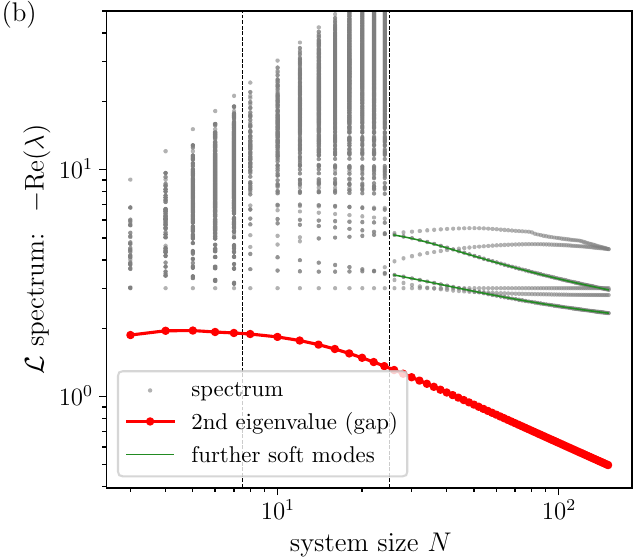}
    \caption{\label{fig:normalCounter-Ising}\textbf{Finite-size scaling of the Liouvillian spectra for the infinite-range transverse-field Ising model.}
     We use field strength $g=1$, spontaneous emission rate $\gamma=1$, and the coupling $J=J_c(\kappa)$ with $\kappa=1/2$ such that $\L$ should be critical and $\L_0$ in the ferromagnetic phase. The thin dashed vertical lines demarcate transitions between exact numerical methods: full $4^N$ Liouville space diagonalization for $N \leq 7$, full block diagonalization in the symmetric subspace for $8 \leq N \leq 24$, and Krylov subspace iterations for the symmetric subspace, targeting the ten eigenvalues with largest real parts for $N \ge 26$.
     (a) As the system size $N$ increases, the second eigenvalue of the unperturbed Liouvillian $\L_0$ plummets exponentially toward zero, dynamically disconnecting the two symmetry-broken steady states. The third eigenvalue, which converges to the true dissipative gap for $N\to\infty$, remains strictly bounded away from zero. (b) The spectrum of the combined Liouvillian $\L = \L_0 + \D_\kappa$ evaluated precisely at the non-equilibrium critical point. The addition of the auxiliary gap-opening dissipator $\D_\kappa$ restores symmetry and drives a continuous transition. Instead of remaining isolated, the dissipative gap (second eigenvalue) and other soft collective modes (green lines) decay algebraically, forming a gapless continuum for $N\to\infty$ and illustrating critical slowing down.}
\end{figure}

This shows again that the gap-opening dissipators of Props.~\ref{prop:normal-qubit}-\ref{prop:normal-s} are generally ineffective
in the presence of non-unital dissipation. In this particularly drastic example, the auxiliary noise $\D_\kappa$ turns out to actually \emph{close} the gap as $J_c(\kappa)$ is the critical point of $\L$ but not $\L_0$.

As $J_c(\kappa) > J_c(0)$, for $J=J_c(\kappa)$, $\L_0$ sits in the ferromagnetic phase of the model and is gapped: In the thermodynamic limit, transitions between the two ferromagnetic steady states are completely suppressed, and the zero eigenvalue ($\lambda=0$) simply acquires a multiplicity of two. Linearizing the equations of motion around one of the asymmetric steady states, one finds that the fluctuations around this local minimum are strongly suppressed. Because the broken symmetry is discrete ($\mathbb{Z}_2$) rather than continuous, Goldstone's theorem does not apply, and all excitations decay at a finite rate.

Figure~\ref{fig:normalCounter-Ising} visualizes this contrasting spectral behavior with finite-size numerics. The $\L_0$ eigenvalue with the second largest real part drops exponentially toward zero as $N$ increases. Deep in this symmetry-broken phase, the effective macroscopic potential features two disjoint basins of attraction. Because the Lindblad operators act locally, transitions between these macroscopic attractors require a sequence of $\mathcal{O}(N)$ local quantum jumps. The associated transition rate -- and thus the second $\L_0$ eigenvalue, which governs the relaxation of the antisymmetric superposition of the two attractors -- is suppressed by a macroscopic action barrier, leading to the observed exponential decay and the dynamical disconnection of the two ferromagnetic steady states. In the thermodynamic limit, the third $\L_0$ eigenvalue embodies the true dissipative gap, characterizing local fluctuations within a given broken-symmetry minimum. As a function of $N$, it remains strictly bounded away from zero.
In stark contrast, when the added auxiliary dissipator restores the symmetry such that $\L=\L_0+\D_\kappa$ is exactly at the non-equilibrium critical point, the dissipative gap (the second eigenvalue) is no longer isolated; alongside all higher collective excitations within the symmetric sector $\mathbb{L}_\text{sym}$, it decays algebraically to zero. This algebraic gap closing is a direct signature of critical slowing down driven by the collective soft modes.

For $N\geq 8$, the numerical analysis was restricted to the permutation-symmetric Liouville subspace 
\begin{equation}
	\mathbb{L}_\text{sym} = \left\{ \dm \in \mathbb{L} \ \mid \ \hU_\pi \dm \hU_\pi^\dag = \dm \ \ \forall \pi \in S_N \right\}.
\end{equation}
The macroscopic collective modes that drive the mean-field dynamics, spontaneous symmetry breaking, and critical slowing down are entirely contained in $\mathbb{L}_\text{sym}$. Excitations outside of this sector represent microscopic relative fluctuations (non-symmetric coherences) between the spins. Because the Hamiltonian is fully permutation-symmetric, it provides no macroscopic protection or collective slowing down for these relative modes. Consequently, they are damped independently by the local dissipators, enforcing a rapid decay governed by the microscopic single-particle relaxation rates with decay rates bounded from below by $\gamma/2$. Therefore, in the thermodynamic limit, the slowest relaxing modes -- and thus the dissipative gap itself -- are determined by the symmetric subspace. The dimension of $\mathbb{L}_\text{sym}$ is given by the number of ways to distribute $N$ identical spins among the $4$ basis states of a single spin's local Liouville space, i.e., $\dim \mathbb{L}_\text{sym}=(N+1)(N+2)(N+3)/6$.

\section{Discussion}\label{sec:discuss}
In this work, we have presented a direct algebraic proof for the fundamental non-positivity of Liouvillian spectral values in Markovian open quantum systems. By working directly with the algebraic properties of the Liouville super-operator $\L$, this approach avoids traditional proofs that rely on the macroscopic contractivity of the quantum channels generated by $\L$.

Our discussion highlights important subtleties that emerge in systems with infinite-dimensional Hilbert spaces. We clarified the vital distinction between the formal algebraic eigenoperators and the rigorous spectrum of the properly defined adjoint Liouvillian $\L^\dag$. While physical instabilities, such as continuous particle pumping, can manifest through unbounded observables possessing positive formal $\L^\dag$ eigenvalues and diverging expectation values, the rigorous spectrum remains strictly confined to the left half-plane when $\L^\dag$ is correctly restricted to the domain of bounded operators. This emphasizes an important distinction: The non-positivity of the Liouvillian spectrum guarantees the probabilistic and informational stability of the open system, ensuring it behaves as a valid quantum channel, but this does not guarantee thermodynamic or energetic stability.

Furthermore, we demonstrated the practical utility of this algebraic framework by establishing rigorous lower bounds for the dissipative gap in Markovian many-body systems. Specifically, we proved that for any unital spin or fermionic many-particle system, one can impose a strict minimum on the dissipative gap by introducing simple single-site dissipators, such as Pauli, generalized Pauli, or Majorana operators.
Counterexamples demonstrate that the unitality constraint on the Lindblad operators is absolutely necessary for these bounds to hold. In systems with non-unital dissipators, such as highly asymmetric jumps, the resulting off-diagonal mixing can overwhelm the homogenizing effect of the auxiliary noise. This was most starkly illustrated in the thermodynamic limit of the infinite-range transverse-field Ising model subject to local spontaneous emission. Rather than opening the gap, the addition of the auxiliary gap-controlling dissipator restored the broken symmetry of the ferromagnetic phase and tuned the system exactly to a non-equilibrium critical point, closing the dissipative gap algebraically due to critical slowing down governed by macroscopic collective soft modes.

Ultimately, these results reveal a rich interplay between the microscopic algebraic properties of Lindblad operators and the macroscopic stability of open many-body systems. The framework developed here provides a basis for future work classifying how specific types of local noise either protect dissipative gaps or drive non-equilibrium phase transitions.

\vspace{0.5em}\noindent \emph{Acknowledgments:}  We gratefully acknowledge support through the U.S.\ National Science Foundation grant no.\ PHY-2412555.

\vspace{0.5em}\noindent \emph{Data availability:} The numerical data shown in Fig.~\ref{fig:normalCounter-Ising} are provided in the repository~\cite{Barthel2026_08data}.

\end{document}